\documentclass[dvipsnames,usenames]{lmcs} 
\pdfoutput=1

\usepackage{lastpage}
\lmcsdoi{15}{1}{25}
\lmcsheading{}{\pageref{LastPage}}{}{}%
{Feb.~20,~2018}{Mar.~06,~2019}{}


\usepackage{hyperref}
\usepackage[utf8]{inputenc}
\usepackage{amssymb,amsmath,amsthm}
\usepackage{url}
\usepackage{xspace}
\usepackage{tikz}
\usetikzlibrary{shapes}
\tikzset{
vertex/.style={circle,draw,inner sep=.5pt,fill=Black}
}

\newcommand{\Case}[2]{\smallskip\par{\it Case #1:\/ #2}}
\newcommand{\Subcase}[2]{\smallskip\par{\it Subcase #1:\/ #2}}

\newcounter{claim}
\renewcommand{\theclaim}{\Alph{claim}}
\newenvironment{claim}{\refstepcounter{claim}%
\par\medskip\par\noindent{\it Claim~\theclaim.~}~\rm}%
{\par\smallskip\par}
\newenvironment{subproof}{\par\noindent{\sl Proof of Claim~\theclaim.~}}%
{$\,\triangleleft$\par\medskip\par}

\newcounter{oq}
\newcommand{\que}{\refstepcounter{oq}\par{\sc \theoq.}~}

\newcommand{\refeq}[1]{(\ref{eq:#1})}

\newcommand{\of}[1]{\left( #1 \right)}
\newcommand{\setdef}[2]{\left\{ \hspace{0.5mm} #1 : \hspace{0.5mm} #2 \right\}}
\newcommand{\feq}{\stackrel{\mbox{\tiny def}}{=}}

\newcommand{\E}{\exists}

\newcommand{\und}{\wedge}

\newcommand{\classc}{\ensuremath{\mathcal C}\xspace}

\newcommand{\isi}[1]{\ensuremath{\mathrm{ISI}(#1)}\xspace}
\newcommand{\tw}{\mathit{tw}}

\newcommand{\Arb}{\mathrm{Arb}}
\newcommand{\compl}[1]{\overline{#1}}

\newcommand{\indsubgr}[1]{\ensuremath{\mathcal I(#1)}\xspace}
\newcommand{\emso}{\ensuremath{{\exists\mathrm{MSO}}}\xspace}
\newcommand{\emsod}{D_{\emso}}

\newcommand{\cclass}[1]{\textsf{\upshape #1}}
\newcommand{\ac}[1]{\cclass{AC$^{\cclass{#1}}$}\xspace}
\newcommand{\envi}{\mathop{\mathit{Env}}\nolimits}
\newcommand{\env}[2]{\envi_{#1}(#2)}
\newcommand{\deco}{\mathop{\mathit{Dec}}\nolimits}
\newcommand{\dec}[2]{\deco_{#1} #2}

\newcommand{\alice}[1]{E[#1]}
\newcommand{\alicek}[1]{E_\kappa[#1]}
\newcommand{\ergraph}{G(n,1/2)}
\newcommand{\ea}[1]{\mathit{EA}_{#1}}
\newcommand{\rturan}{\mathbb{T}_{k,n}}
\newcommand{\fo}{\ensuremath{\mathrm{FO}}\xspace}

\begin{document}

\title[On the First-Order Complexity of  Induced Subgraph Isomorphism]{On the First-Order Complexity of \\ Induced Subgraph Isomorphism\rsuper*}
\titlecomment{{\lsuper*}A preliminary version of this paper appeared in [28]. Theorem 6.2 was announced without proof in~[27].}

\author[O. Verbitsky]{Oleg Verbitsky}%
\address{Institut f\"ur Informatik,  Humboldt-Universit\"at zu Berlin, Unter den Linden 6, D-10099 Berlin.}
\email{verbitsk@informatik.hu-berlin.de}
\thanks{  Supported by DFG grant VE 652/1--2. On leave from the IAPMM, Lviv, Ukraine.}

\author[M. Zhukovskii]{Maksim Zhukovskii}%
\address{Laboratory of Advanced Combinatorics and Network Applications,
Moscow Institute of Physics and Technology, Moscow.}
\email{zhukmax@gmail.com}
\thanks{ Supported by grants No.\ 15-01-03530 and 16-31-60052 of Russian Foundation for Basic Research.}









\begin{abstract}
Given a graph $F$, let $\indsubgr F$ be the class of graphs containing $F$
as an induced subgraph. Let $W[F]$ denote the minimum $k$ such that
$\indsubgr F$ is definable in $k$-variable first-order logic. The recognition
problem of $\indsubgr F$, known as Induced Subgraph Isomorphism (for the pattern
graph $F$), is solvable in time $O(n^{W[F]})$. Motivated by this fact,
we are interested in determining or estimating the value of $W[F]$.
Using Olariu's characterization of paw-free graphs, we show that $\indsubgr{K_3+e}$
is definable by a first-order sentence of quantifier depth 3, where
$K_3+e$ denotes the paw graph. This provides an example of a graph $F$
with $W[F]$ strictly less than the number of vertices in $F$.
On the other hand, we prove that $W[F]=4$ for all $F$ on 4 vertices
except the paw graph and its complement. If $F$ is a graph on $\ell$ vertices,
we prove a general lower bound $W[F]>(1/2-o(1))\ell$, where the function
in the little-o notation approaches 0 as $\ell$ increases. This bound
holds true even for a related parameter $W^*[F]\le W[F]$, which is
defined as the minimum $k$ such that $\indsubgr F$ is definable in the
infinitary logic $L^k_{\infty\omega}$. We show that $W^*[F]$ can be
strictly less than $W[F]$. Specifically, $W^*[P_4]=3$ for $P_4$ being
the path graph on 4 vertices.

Using the lower bound for $W[F]$, we also obtain a succintness
result for existential monadic second-order logic:
a single monadic quantifier sometimes reduces
the first-order quantifier depth at a super-recursive rate.
\end{abstract}

\maketitle

\section{Introduction}\label{s:intro}

For a given graph $F$, let \indsubgr F denote the class of all graphs containing a copy of $F$
as an induced subgraph. We are interested in the descriptive complexity of \indsubgr F,
for a fixed pattern graph $F$, in first-order logic whose vocabulary consists of
the adjacency and the equality relations ($\sim$ and $=$ respectively).
Let $D[F]$ denote the minimum quantifier depth of a sentence in this logic
that defines  \indsubgr F. Furthermore, let $W[F]$ denote the minimum variable width
of a sentence defining  \indsubgr F, that is, the minimum number of variables
in such a sentence, where a variable with multiple occurrences counts only once. Note that
$$
W[F]\le D[F]\le\ell,
$$
where here and throughout the paper $\ell$ denotes the number of vertices in $F$.
It may come as some surprise that the parameter $D[F]$ can be strictly less than $\ell$.
To see an example, let $F=K_3+e$ be the paw graph
\begin{tikzpicture}[every node/.style=vertex,scale=.25]
\path (0,0) node (a) {}
      (0,1) node (b)  {} edge (a) 
      (.5,.5) node (c)  {} edge (a) edge (b)
      (1.2,.5) node (d)  {} edge (c);
\end{tikzpicture}.
The following sentence uses three variables $x_1,x_2,x_3$ and has quantifier depth~3:
\begin{multline*}
\E x_1\,(\E x_2\E x_3\,(x_1\sim x_2\und x_1\sim x_3\und x_2\sim x_3) \und {}\\
 \E x_2\,(
x_1\not\sim x_2\und
\E x_3\,(x_1\sim x_3\und x_3\sim x_2)\und
\E x_3\,(x_3\sim x_1\und x_3\not\sim x_2)
)
). 
\end{multline*}
It says that a graph contains a vertex $v$ that belongs to a triangle and
can be accompanied with a vertex $u$ at distance 2 from $v$ such that
there is a vertex $w$ adjacent to $v$ but non-adjacent to $u$.
Obviously, this sentence is true on the paw graph and on every graph
containing an induced paw subgraph. 
Olariu's characterization of paw-free graphs \cite{Olariu88} implies
that the sentence is false on every graph without an induced paw subgraph;
see Section \ref{ss:paw} for details.  

The decision problem for \indsubgr F is known as
\textsc{Induced Subgraph Isomorphism} (for the pattern graph $F$).
We denote it by \isi F. Our interest in the parameters $D[F]$ and $W[F]$
is motivated by the fact that \isi F is solvable in time $O(n^{W[F]})$;
see \cite[Prop.~6.6]{Libkin04}. Before stating our results on $D[F]$ and $W[F]$
we give a brief overview over the known algorithmic results in this area.

\subsection*{Computational complexity of Induced Subgraph Isomorphism.}
Obviously, \isi F is solvable in time $O(n^\ell)$
on $n$-vertex input graphs by exhaustive search.
We use the standard notation
$K_\ell$ for complete graphs, 
$P_\ell$ for paths, and $C_\ell$ for cycles on $\ell$ vertices. 
Itai and Rodeh \cite{ItaiR78} observed that \isi{K_3} is solvable
in time $O(n^\omega)$, where $\omega<2.373$ is the exponent of fast square
matrix multiplication~\cite{Gall14}.
Ne\v{s}et\v{r}il and Poljak \cite{NesetrilP85} showed, by a
reduction of \isi F to \isi{K_3}, 
that \isi F is solvable in time $O(n^{(\omega/3)\ell+c})$,
where $c=0$ if $\ell$ is divisible by 3 and $c\le\frac23$ otherwise.
For $\ell$ not divisible by 3, this time bound was improved by Eisenbrand and Grandoni \cite{EisenbrandG04}
using fast rectangular matrix multiplication.
On the other hand, \isi{K_\ell} is unsolvable in time $n^{o(\ell)}$
unless the Exponential Time Hypothesis fails \cite{ChenHKX06}.
Floderus et al.~\cite{FloderusKLL15}
collected evidence in favour of the conjecture that \isi F for $F$ with $\ell$
vertices cannot be solved faster than \isi{K_{c\ell}}, where $c<1$ is a constant independent on $F$.
Along with the Exponential Time Hypothesis, this would imply that
the time complexity of \isi F is $n^{\Theta(\ell)}$.
As an example of a particular result of \cite{FloderusKLL15} in this direction, 
\isi{P_{2a-1}} is not easier than \isi{K_a},
and the same holds true for \isi{C_{2a}}; see also the earlier work~\cite{ChenF07}.

The induced subgraph isomorphism problem has been intensively studied for 
particular pattern graphs $F$ with small number of vertices.
Let $\compl F$ denote the complement graph of $F$ and note that at least one
of the graphs $F$ and $\compl F$ is connected.
Since \isi F and \isi{\compl F} have the same time complexity,
we can restrict our attention to connected pattern graphs.
There are six such graphs on four vertices,
namely $K_4$, $P_4$, $C_4$, $K_3+e$, the claw graph $K_{1,3}$
(
\begin{tikzpicture}[every node/.style=vertex,scale=.25]
\path (0,0) node (a) {}
      (-.5,1) node (b)  {} edge (a) 
      (0,1) node (c)  {} edge (a)
      (.5,1) node (d)  {} edge (a);
\end{tikzpicture}
),
and the diamond graph $K_4\setminus e$
(
\begin{tikzpicture}[every node/.style=vertex,scale=.25]
\path (0,0) node (a) {}
      (0,1) node (b)  {}
      (-.5,.5) node (c)  {} edge (a) edge (b)
      (.5,.5) node (d)  {} edge (a) edge (b) edge (c);
\end{tikzpicture}
).
Corneil et al.~\cite{CorneilPS85} designed an $O(n+m)$ time algorithm for \isi{P_4},
where $m$ denotes the number of edges in an input graph. As noted in \cite{KloksKM00},
the Olariu characterization of paw-free graphs reduces \isi{K_3+e} to \isi{K_3},
showing that the former problem is also solvable in time $O(n^\omega)$.
The same time bound is obtained by Vassilevska Williams et al.~\cite{WilliamsWWY15}
for the diamond graph $K_4\setminus e$ and, using randomization, for the other
pattern graphs on 4 vertices except $K_4$.
The best known time bound for \isi{K_4} is given by the methods of~\cite{EisenbrandG04}
and is currently $O(n^{3.257})$~\cite{WilliamsWWY15}.

\subsection*{Our results on the descriptive complexity of Induced Subgraph Isomorphism.}
In Section \ref{s:lower} we prove a general lower bound $W[F]>(1/2-o(1))\ell$, where the function
in the little-o notation approaches 0 as $\ell$, the number of vertices
in the pattern graph $F$, increases.
Whether or not it can be improved remains an intriguing open question.
Note that this bound leaves a hypothetical possibility that
the time bound $O(n^{W[F]})$ for Induced Graph Isomorphism can be better
than the Ne\v{s}et\v{r}il-Poljak bound $O(n^{(\omega/3)\ell+c})$ 
for infinitely many pattern graphs~$F$.

Our approach uses a connection to the $k$-extension property of graphs, that is well
known in finite model theory; see, e.g, \cite{Spencer-book}.
We define the \emph{extension index} of $F$, denoted by $\alice F$, as
the minimum $k$ such that the $k$-extension property forces the existence
of an induced copy of $F$. It is easy to show that $W[F]\ge \alice F$. 
Our results about the parameter $\alice F$ may be interesting on its own.
In particular, we show that $\alice F\ge\chi(F)$, where $\chi(F)$
denotes the chromatic number of~$F$.

In Section \ref{s:small}, we determine the values of $D[F]$ and $W[F]$
for all pattern graphs with at most 4 vertices.
We prove that $W[F]=4$ for all $F$ on 4 vertices
except the paw graph and its complement.

With one exception, our lower bounds for $W[F]$ hold true 
also for a related parameter $W^*[F]\le W[F]$, which is
defined as the minimum $k$ such that $\indsubgr F$ is definable in the
infinitary logic $L^k_{\infty\omega}$. 
For the exceptional pattern graph $F=P_4$, the path on 4 vertices,
we prove that $W^*[P_4]=3$ while $W[P_4]=4$. 
This shows that $W^*[F]$ can be
strictly less than $W[F]$, that is, the infinitary logic can be
more succinct when defining~$\indsubgr F$.

In Section \ref{s:highly}, we address a relaxation version of the parameter $W[F]$.
Consider a simple example. Let $D_v[F]$ be the minimum quantifier depth of a sentence $\Phi$
defining \indsubgr F over \emph{sufficiently large connected} graphs. That is,
it is required that there is a number $s$ such that $\Phi$ correctly detects whether or 
not a graph $G$ belongs to \indsubgr F only if $G$ is connected and has at least $s$ vertices.
Whereas $D_v[F]\le D[F]$, it is clear that \indsubgr F for a connected pattern graph $F$ is still recognizable
in time $O(n^{D_v[F]})$. Let $F=P_3$. As easily seen,
$P_3$-free graphs are exactly disjoint unions of cliques.
Therefore, connected $P_3$-free graphs are exactly the complete graphs, which
readily implies that $D_v[P_3]\le2$, whereas $D[P_3]=3$.
As a further example, we remark that the existence of a \emph{not necessarily induced}
subgraph $P_4$ can be defined over sufficiently large connected graphs with just 2 variables;
see \cite{VZh16} for this and further examples.

We can go further and define $W_\tw[F]$ to be the minimum variable width  of a sentence
defining \indsubgr F over connected graphs $G$ of \emph{sufficiently large treewidth} $\tw(G)$.
As a consequence of Courcelle's theorem \cite{Courcelle90}, \indsubgr F for a connected pattern graph $F$ 
is recognizable in time $O(n^{W_\tw[F]})$; cf.\ the discussion in~\cite{VZh16}.

The above discussion motivates the problem of proving lower bounds for the parameter
 $W_\kappa[F]$ which we define as the minimum variable width  of a sentence
defining \indsubgr F over graphs $G$ of \emph{sufficiently large connectedness} $\kappa(G)$.
Note that $W_\kappa[F]\le W_\tw[F]\le W_v[F]\le W[F]$. We prove that $W_\kappa[F]=W[F]$
for a large class of pattern graphs $F$. We also prove a general lower bound
$W_\kappa[F]>(1/3-o(1))\,\ell$ for all $F$ on $\ell$ vertices.

Finally, we notice that the lower bound $W[F]=\Omega(\ell)$ implies a
succintness result for existential monadic second-order logic:
A usage of just one monadic quantifier sometimes reduces
the first-order quantifier depth at a super-recursive rate.
More precisely, let $\emsod[F]$ denote the minimum quantifier depth
of a second-order sentence with a single existential monadic quantifier
that defines \indsubgr F. In Section \ref{s:emso}, we prove that 
$\emsod[F]$ can sometimes be so small compared to $D[F]=D_\fo[F]$
that there is no total recursive function $f$ such that
$f(\emsod[F])\ge D[F]$ for all~$F$.

\subsection*{Comparison to (not necessarily induced) Subgraph Isomorphism.}
The \textsc{Subgraph Isomorphism} problem
is very different from its induced version.
For infinitely many pattern graphs $F$, \textsc{Subgraph Isomorphism} can be solved in time $O(n^c)$
for a constant $c$. This follows from a result by Alon, Yuster and
Zwick~\cite{AlonYZ95} who showed that the problem is solvable in time
$2^{O(\ell)}\cdot n^{\tw(F)+1}\log n$.

Let $W(F)$ and $D(F)$ be the analogs of $W[F]$ and $D[F]$ for the not-necessarily-induced case.
Note that $W(F)=D(F)=\ell$ as a consequence of the trivial observation that $K_\ell$ contains
$F$ as a subgraph while $K_{\ell-1}$ does not. Nevertheless, $D_v(F)$ can be strictly smaller
than $\ell$ for some connected pattern graphs. Moreover, $D_\tw(F)$ can sometimes be arbitrarily
small compared to~$\ell$. This is the subject of our preceding paper~\cite{VZh16}.
Furthermore, in \cite{VZh18} we have shown that $D_v(F)\le\frac23\,\ell+1$  
for infinitely many $F$. This upper bound is tight as, on the other hand,
we have $W_v(F)>\frac23\,\ell-2$ for every~$F$.

\subsection*{Logic with numeric predicates.}
In the present paper, we consider the most laconic first-order language of graphs
whose vocabulary has only the adjacency and the equality relations.
If we assume that the vertex set of a graph is $\{1,2,\ldots,n\}$ and
additionally allow arbitrary numerical relations
like order, parity etc., this richer logic captures non-uniform \ac0; see \cite{Immerman-book,Libkin04}.
Let $W_\Arb(F)$ denote the analog of the parameter $W(F)$ (the not-necessarily-induced case)
for this logic,
and $W_\Arb[F]$ denote the analog of the parameter $W[F]$ (the induced case).
The known relations to circuit complexity \cite{Immerman-book,Rossman08} imply that
the (not necessarily induced) \textsc{Subgraph Isomorphism} is solvable by bounded-depth unbounded-fan-in circuits of
size $n^{W_\Arb(F)+o(1)}$, and the similar bound is true also for \textsc{Induced Subgraph Isomorphism}.
The parameter $W_\Arb(F)$ is studied in this context by Li, Razborov, and Rossman~\cite{LiRR14}.

\section{Preliminaries}\label{s:prel}

A \emph{graph property} is a class of graphs \classc closed under isomorphism, that is,
for isomorphic graphs $G$ and $H$, $G\in\classc$ iff $H\in\classc$.
We consider first-order sentences about graphs in the language containing
the adjacency and the equality relations. A sentence $\Phi$ \emph{defines}
a graph property \classc if $G\in\classc$ exactly when $G\models\Phi$, i.e.,
$\Phi$ is true on $G$.
A graph property \classc is \emph{first-order definable} if there is a first-order
sentence defining \classc.

Let \classc be a first-order graph property. The \emph{logical depth} of \classc,
denoted by $D(\classc)$, is the minimum quantifier depth (rank) of a sentence defining $\classc$.
The \emph{logical width} of \classc,
denoted by $W(\classc)$, is the minimum \emph{variable width} of a sentence defining $\classc$,
i.e., the number of first-order variables occurring in the sentence where different
occurrences of the same variable count only once.

Given two non-isomorphic graphs $G$ and $H$, let
$D(G,H)$ (resp.\ $W(G,H)$) denote the minimum quantifier depth (resp.\ variable width) of a sentence
that \emph{distinguishes} $G$ and $H$, i.e., is true on one of the graphs and false on the other.

\begin{lem}\label{lem:DDGH}
$D(\classc)=\max\setdef{D(G,H)}{G\in\classc,\,H\notin\classc}$.  
\end{lem}

\begin{proof}
In one direction, note that whenever $G\in\classc$ and $H\notin\classc$,
we have $D(G,H)\le D(\classc)$ because any sentence defining $\classc$
distinguishes $G$ and $H$. For the other direction, suppose that
every such $G$ and $H$
are distinguished by a sentence $\Phi_{G,H}$ of quantifier depth at most $d$.
Specifically, suppose that $\Phi_{G,H}$ is true on $G$ and false on $H$.
We have to show that $\classc$ can be defined by a sentence of quantifier depth at most $d$.
For a graph $G\in\classc$,
consider the sentence $\Phi_G\feq\bigwedge_{H}\Phi_{G,H}$,
where the conjunction is over all $H\notin\classc$. 
This sentence distinguishes $G$
from all $H\notin\classc$ and has quantifier depth at most $d$.
The only problem with it is that the conjunction over $H$ is actually infinite.
Luckily, there are only finitely many pairwise
inequivalent first-order sentences about graphs
of quantifier depth $d$; see, e.g., \cite[Theorem 2.4]{PikhurkoV11}.
Removing all but one formula $\Phi_{G,H}$ from each equivalence class,
we make $\Phi_G$ a legitimate finite sentence.
Now, consider $\Phi\feq\bigvee_{G}\Phi_{G}$,
where the disjunction is over all $G\in\classc$.
It can be made finite in the same way.
The sentence $\Phi$ defines $\classc$ and has quantifier depth at most~$d$.
\end{proof}

Thus, Lemma \ref{lem:DDGH} is a simple consequence of the fact that
there are only finitely many pairwise inequivalent first-order statements
of bounded quantifier depth. Note that the last fact does not hold true
for the variable width. We define
$$
W^*(\classc)=\max\setdef{W(G,H)}{G\in\classc,\,H\notin\classc}.
$$
Equivalently, $W^*(\classc)$ is equal to the minimum $k$ such that \classc is
definable in the infinitary logic $L^k_{\infty\omega}$; see, e.g., \cite[Chapter~11]{Libkin04}.
Obviously, $W^*(\classc)\le W(\classc)$, and we will see in Section \ref{ss:path}
that this inequality can be strict. Summarizing, we have
\begin{equation}
  \label{eq:WWD}
W^*(\classc)\le W(\classc)\le D(\classc).  
\end{equation}

The value of $W(\classc)$ admits the following characterization.
If $W(G,H)\le k$, let $D^k(G,H)$ denote the minimum quantifier depth of a first-order $k$-variable sentence
distinguishing $G$ and $H$. We set $D^k(G,H)=\infty$ if $W(G,H) > k$.
The following fact can be proved similarly to Lemma~\ref{lem:DDGH}.

\begin{lem}\label{lem:WDGH}
$W(\classc)>k$ if and only if there is a sequence of graph pairs $(G_i,H_i)$ with $G_i\in\classc$ and 
$H_i\notin\classc$ such that $D^k(G_i,H_i)\to\infty$ as $i\to\infty$.
\end{lem}

Lemmas \ref{lem:DDGH} and \ref{lem:WDGH} reduce estimating the logical depth and width to estimating 
the parameters $D(G,H)$ and $D^k(G,H)$ over $G\in\classc$ and $H\notin\classc$.
The first inequality in \refeq{WWD} can be used for obtaining lower bounds for $W(\classc)$
by estimating $W(G,H)$ over $G\in\classc$ and $H\notin\classc$.
The parameters $D(G,H)$, $D^k(G,H)$, and $W(G,H)$ have a very useful combinatorial characterization.

In the \emph{$k$-pebble Ehrenfeucht-Fra{\"\i}ss{\'e} game} (see, e.g., \cite[Chapter 11.2]{Libkin04}), 
the board consists of two vertex-disjoint graphs $G$ and $H$. 
Two players, \emph{Spoiler} and \emph{Duplicator} (or \emph{he} and \emph{she})
have equal sets of $k$ pairwise different pebbles.
In each round, Spoiler takes a pebble and puts it on a vertex in $G$ or in $H$; 
then Duplicator has to put her copy of this pebble on a vertex
of the other graph.
Note that the pebbles can be reused and change their positions during the play.
Duplicator's objective is to ensure that the pebbling determines a partial
isomorphism between $G$ and $H$ after each round; when she fails, she immediately loses.
The proof of the following facts can be found in~\cite{Immerman-book}:

\begin{enumerate}
\item 
$D^k(G,H)$ is equal to the
minimum $d$ such that Spoiler has a winning strategy in the $d$-round $k$-pebble
game on $G$ and~$H$.
\item 
 $D(G,H)$ is equal to the
minimum $d$ such that Spoiler has a winning strategy in the $d$-round $d$-pebble
game on $G$ and~$H$. 
\item 
 $W(G,H)$ is equal to the
minimum $k$ such that, for some $d$, Spoiler has a winning strategy in the $d$-round $k$-pebble
game on $G$ and~$H$. 
\end{enumerate}

We are interested in the property of containing a specified induced subgraph.
We write $F\sqsubset G$ to say that $G$ contains an induced subgraph isomorphic to $F$.
Thus, $\indsubgr F=\setdef{G}{F\sqsubset G}$.
Let $D[F]=D(\indsubgr F)$ and, similarly, $W[F]=W(\indsubgr F)$ 
and $W^*[F]=W^*(\indsubgr F)$. Thus,
$W^*[F]$ is the maximum $W(G,H)$ over all $G$ containing an induced copy of $F$
and all $H$ not containing such a copy.
As a particular case of \refeq{WWD}, we have
$$
W^*[F]\le W[F]\le D[F]\le\ell
$$
for every $F$ with $\ell$ vertices.

The vertex set of a graph $G$ will be denoted by $V(G)$.
Throughout the paper, we consider simple undirected graphs without loops.
Let $\compl G$ denote the complement of $G$, that is, $V(\compl G)=V(G)$
and two vertices are adjacent in $\compl G$ exactly when they are not adjacent in~$G$.

\begin{lem}\label{lem:compl}
$D[F]=D[\compl F]$, $W^*[F]=W^*[\compl F]$, and $W[F]=W[\compl F]$.
\end{lem}

\begin{proof}
The first equality follows from the equality $D(G,H)=D(\compl G,\compl H)$ by Lemma \ref{lem:DDGH}.
Indeed, $F\sqsubset G$ iff $\compl F\sqsubset\compl G$. Therefore,
\begin{multline*}
D[F]=\max\setdef{D(G,H)}{G\sqsupset F,\,H\not\sqsupset F}=
\max\setdef{D(\compl G,\compl H)}{\compl G\sqsupset \compl F,\,\compl H\not\sqsupset \compl F}\\=
\max\setdef{D(G,H)}{G\sqsupset\compl F,\,H\not\sqsupset\compl F}=D[\compl F].
\end{multline*}

The second equality follows similarly from the equality $W(G,H)=W(\compl G,\compl H)$ by
the definition of $W^*[F]$. The third equality follows from the equality 
$D^k(G,H)=D^k(\compl G,\compl H)$ by Lemma~\ref{lem:WDGH}.
\end{proof}

\subsection*{Further graph-theoretic definitions.}
A graph $G$ is called $F$-free if $F\not\sqsubset G$.
The vertex-disjoint union of graphs $G$ and $H$ will be denoted by $G+H$.
Correspondingly, $sG$ is the vertex-disjoint union of $s$ copies of $G$.
The \emph{lexicographic product} $A\cdot B$ of two
graphs $A$ and $B$ is defined as follows: $V(A\cdot B)=V(A)\times V(B)$, and $(u,v)$ and $(x,y)$ are adjacent
in $A\cdot B$ if $u$ and $x$ are adjacent in $A$ or if $u=x$ and $v$ and $y$
are adjacent in $B$. In other words, $A\cdot B$ is obtained from $A$
by substituting each vertex $u$ with an induced copy $B_u$ of $B$
and drawing all edges between $B_u$ and $B_x$ whenever $u$ and $x$ are adjacent.

A vertex is \emph{isolated} if it has no adjacent vertex
and \emph{universal} if it is adjacent to all other vertices in the graph.
Two vertices are called \emph{twins} if they have the same adjacency to
the rest of the graph.

Throughout the paper, $\log n$ means the logarithm base~2.

\section{The extension index and a lower bound for $W^*[F]$}\label{s:lower}

Let $k\ge2$.
By the \emph{$k$-extension property} we mean the first-order sentence $\ea k$ of quantifier depth $k$
(also called the \emph{$k$th extension axiom})
saying that, for every two disjoint sets $X,Y\subset V(G)$ with $|X\cup Y| < k$, 
there is a vertex $z\notin X\cup Y$ adjacent to all $x\in X$ and non-adjacent to all $y\in Y$.
Note that $\ea 2$ says exactly that a graph has
neither isolated nor universal vertex.
For convenience, we also set $\ea 1\feq\E z(z=z)$.

Note that, if $G\models \ea k$ and $F$ has at most $k$ vertices, then $F\sqsubset G$.
Suppose that $F$ has more than 1 vertex.
We define the \emph{extension index} of $F$, denoted by $\alice F$, as
the minimum $k$ such that $H\models \ea k$ implies $F\sqsubset H$.
Equivalently, $\alice F$ is
the maximum $k$ for which there is a graph $H$ such that
$H\models \ea{k-1}$ while $F\not\sqsubset H$. 
Note that $\alice F\le\ell$ for any $\ell$-vertex graph~$F$.

\begin{lem}\label{lem:alice}
$W^*[F]\ge \alice F$.  
\end{lem}

\begin{proof}
As easily seen, if both $G$ and $H$ have the $(k-1)$-extension property, then 
Duplicator has a winning strategy in the $(k-1)$-pebble Ehrenfeucht-Fra{\"\i}ss{\'e} 
game on graphs $G$ and $H$ and, hence, $W(G,H)\ge k$. 
Therefore, it suffices to show that there are $G$ and $H$ such that $F\sqsubset G$, 
$F\not\sqsubset H$, and both of them satisfy $\ea{k-1}$ for $k=\alice F$. Such a graph $H$ exists by the
definition of the extension index. Such a graph $G$ exists because,
as very well known (see, e.g., \cite{Spencer-book}), for fixed $k$ and $\ell$ 
a random graph $\ergraph$ has 
the $k$-extension property and contains every $\ell$-vertex graph as an induced subgraph
with probability approaching 1 as $n$ increases.
Recall that $G(n,p)$ refers to the probability distribution on graphs with
vertex set $\{1,\ldots,n\}$ where each two vertices are adjacent with probability $p$
independently of the other pairs.
\end{proof}

\begin{exa}\label{ex:3-vertex}\hfill
  \begin{enumerate}[\bf 1.]
  \item 
$\alice{P_3}=3$ because $H=2K_2$ is $P_3$-free and satisfies $\ea 2$.
By Lemma \ref{lem:alice}, $W^*[P_3]=W[P_3]=3$. 
\item 
$\alice{K_3}=3$, as also certified by $H=2K_2$ (or by $H=C_4$).
  \end{enumerate}
\end{exa}

We can determine $\alice{K_\ell}$ for any $\ell$
using a relationship between $\alice F$ and the chromatic number of~$F$.

\begin{thm}\label{thm:chi}
$\alice F\ge\chi(F)$.
\end{thm}

\begin{proof}
Let $k=\chi(F)-1$. We have to show that there is a graph $G$ having the $k$th
extension property and containing no induced copy of $F$.

Let $T_{k,n}$ denote the $k$-partite Turán graph with $kn$ vertices.
The vertex set of $T_{k,n}$ is split into $k$ vertex classes $V_1,\ldots,V_k$,
each consisting of $n$ vertices. Two vertices of $T_{k,n}$ are adjacent if and
only if they belong to different vertex classes. Obviously, $\chi(T_{k,n})=k$.
Since $\chi(T_{k,n})<\chi(F)$, the graph $T_{k,n}$ itself and any of its subgraphs
do not contain an induced copy of~$F$. Let $\rturan$ be a random subgraph of $T_{k,n}$,
obtained from $T_{k,n}$ by deleting each edge with probability $1/2$, independently
of the other edges. In order to prove the theorem, it suffices to show that
$\rturan$ has the $k$th extension property with nonzero probability if $n$
is chosen sufficiently large.

Consider two disjoint vertex sets $X,Y$ in $\rturan$ such that $|X\cup Y|=k-1$
and estimate the probability that they violate $\ea k$.
Fix a vertex class $V_m$ disjoint with $X\cup Y$.
A particular vertex $z\in V_m$ is adjacent to all $x\in X$ and non-adjacent to all $y\in Y$
with probability $2^{-k+1}$, and the converse happens with probability $1-2^{-k+1}$.
The probability that none of the vertices in $V_m$ has the ``right'' adjacency pattern
to $X$ and $Y$ is equal to $(1-2^{-k+1})^n$. 
Using the inequalities
${a\choose b}<\of{\frac{a\,\mathrm{e}}b}^b$ and
\begin{equation}
  \label{eq:exp}
  1-x\le\mathrm{e}^{-x}\text{ for all reals }x,
\end{equation}
we conclude that two sets $X,Y$
violating $\ea k$ exist with probability at most
$$
{kn\choose k-1}2^{k-1}(1-2^{-k+1})^n\le\of{\frac{kn\,\mathrm{e}}{k-1}}^{k-1}2^{k-1}\mathrm{e}^{-2^{-k+1}n}=
\of{\frac{2k\,\mathrm{e}}{k-1}}^{k-1}\mathrm{e}^{(k-1)\ln n-2^{-k+1}n},
$$
which approaches $0$ as $n$ increases (since $k$ is fixed).
It follows that $\ea k$ is violated by $\rturan$ with probability strictly less
than $1$ if $n$ is chosen sufficiently large.
\end{proof}

\begin{cor}\label{cor:complete}
$\alice{K_\ell}=\ell$.  
\end{cor}

It may seem plausible at first glance that $\alice F=\ell$ for every $F$ with $\ell$ vertices.
Nevertheless, in Section \ref{s:small} we will see that this is not always the case as,
for example, $\alice F=3$ for $F$ being the paw and the path on 4 vertices.
The best general lower bound for $\alice F$ we can show is given by the following lemma.

\begin{lem}\label{lem:alice-lower}
Let $F$ be a graph with $\ell\ge2$ vertices. Then
$$
\alice F\ge\lfloor\frac12\,\ell-2\log_2\ell+3\rfloor.
$$
\end{lem}

\begin{proof}
The lemma is trivially true if $\ell\le15$ because in this case it just states that $\alice F\ge2$.
We, therefore, suppose that $\ell\ge16$.

Denote $k=\lfloor\frac12\,\ell-2\log \ell+2\rfloor$.
Suppose that $\ell$ is even and set $n=2^{\ell/2-1}$.
It suffices to show that the random graph $\ergraph$
with a non-zero probability has the $k$-extension property
and simultaneously contains no induced copy of~$F$.

The probability of $F\sqsubset\ergraph$ is bounded from above by the value of
$$
p(\ell,n)=n(n-1)(n-2)\cdots(n-\ell+1)2^{-\ell(\ell-1)/2}.
$$
By the choice of $n$,
$$
2^{-\ell(\ell-1)/2}=(2n)^{-\ell+1}=2^{-\ell+1}n^{-\ell+1}=\frac12\,n^{-\ell-1}.
$$
Therefore,
$$
p(\ell,n)=\frac{n(n-1)(n-2)\cdots(n-\ell+1)}{2\,n^{\ell+1}}<\frac1{2n}.
$$
It remains to prove that $\ergraph$ has the $k$-extension property
with probability at least~$1/(2n)$.

The probability of $\ergraph\not\models \ea k$ is bounded from above by the value of
$$
q(n,k)={n \choose k-1}2^{k-1}(1-2^{-k+1})^{n-k+1}.
$$
In its turn,
$$
q(n,k) < 4\,n^{k-1}(1-2^{-k+1})^n.
$$
By \refeq{exp}, the last value is bounded from above by
$$
q'(n,k)=4\,\exp\of{\ln n(k-1)-n\,2^{-k+1}}.
$$
Denote $k'=\frac12\,\ell-2\log\ell+2$.
Since the function $f(x)=\ln n\, x-n\,2^{-x}$ is monotonically increasing,
\begin{multline*}
q(n,k)<q'(n,k')=4\,n^{k'-1}\exp\of{-\frac n{2^{\ell/2-2\log\ell+1}}}=
4\,n^{k'-1}\exp\Big(-\frac{\ell^2}4\Big)\\[2mm]
=4\,n^{k'-1}(2n)^{-\log\mathrm{e}\,\ell/2}=
4\,n^{k'-1}2^{-\log\mathrm{e}\,\ell/2}n^{-\log\mathrm{e}\,\ell/2}=
4\,n^{k'-1}(2n)^{-\log\mathrm{e}}n^{-\log\mathrm{e}\,\ell/2}\\[1.5mm]
=4\,\mathrm{e}^{-1}n^{-\ell(\log\mathrm{e}-1)/2-2\log\ell-\log\mathrm{e}+1}.
\end{multline*}
For $\ell\ge16$, this gives us $q(n,k)<2\,n^{-11}$.
Therefore, $\ergraph$ has the $k$-extension property
with probability more than $1-2\,n^{-11}$. This is well more than $1/(2n)$, as desired.

If $\ell$ is odd, set $n=2^{(\ell-3)/2}$ and proceed similarly to above.
\end{proof}

\begin{rem}
The bound of Lemma \ref{lem:alice-lower} cannot be much improved as long as
the argument is based on $\ergraph$. Indeed, it is known \cite{JLR-book}
that there is a function $\ell_0(n)=2\log n-2\log\log n+\Theta(1)$
such that the clique number of $\ergraph$ is equal to $\ell_0(n)$ or to $\ell_0(n)+1$
with probability $1-o(1)$. In \cite{KimPSV05} it is shown that there is a function
$k(n)=\log n-2\log\log n+\Theta(1)$ such that, with probability $1-o(1)$,
$\ergraph$ satisfies $\ea{k(n)}$ but does not satisfy $\ea{k(n)+6}$.
It follows that, if $n$ is chosen so that $\ergraph$ does not contain
a subgraph $K_\ell$ with high probability, then $\ergraph$ satisfies $\ea k$ 
with non-negligible probability for, at best, $k=\frac12\ell-\log\ell+\Theta(1)$. 
\end{rem}

As an immediate consequence of Lemmas \ref{lem:alice} and \ref{lem:alice-lower}
we obtain the following result.

\begin{thm}\label{thm:W^*[F]}
Let $F$ be a graph with $\ell\ge2$ vertices. Then
$$
W^*[F]>\frac12\,\ell-2\log_2\ell+2.
$$
\end{thm}

\section{Four-vertex subgraphs}\label{s:small}

Our next goal is to determine the values of $D[F]$, $W[F]$, and $W^*[F]$
for all graphs $F$ with at most 4 vertices. It is enough to consider
connected $F$, as follows from Lemma \ref{lem:compl} and the fact that 
the complement of a disconnected graph is connected.
The two connected 3-vertex graphs are considered in Example \ref{ex:3-vertex},
and we now focus on connected graphs with 4 vertices.
Recall that $W^*[K_4]=4$ by Corollary~\ref{cor:complete}
(or just because $W(K_4,K_3)=4$).

\subsection{The paw subgraph ($K_3+e$)}\label{ss:paw}

\begin{lem}[Olariu \cite{Olariu88}]\label{lem:olariu}
A graph $H$ is paw-free if and only if each connected component of $H$
is triangle-free or complete multipartite.
\end{lem}

Note that a graph $B$ is complete multipartite iff the complement of $B$
is a vertex-disjoint union of complete graphs. The latter condition means exactly that
$\compl B$ is $P_3$-free. Thus, $B$ is complete multipartite iff
it is $(K_2+K_1)$-free, where $K_2+K_1=\compl{P_3}$.

\begin{thm}\label{thm:paw}
$D[K_3+e]=W[K_3+e]=W^*[K_3+e]=\alice{K_3+e}=3$.
\end{thm}

\begin{proof}
We have $\alice{K_3+e}>2$ because, for example, $C_4$ satisfies the 2nd extension
axiom and does not contain $K_3+e$ as a subgraph.

In order to prove that $D[K_3+e]\le3$, we have to describe a winning strategy for Spoiler in
the 3-round Ehrenfeucht-Fra{\"\i}ss{\'e} game on graphs $G\sqsupset K_3+e$ and $H\not\sqsupset K_3+e$.
Let $v_1,v_2,v_3,v_4$ be vertices spanning a paw in $G$.
We suppose that $v_1$ and $v_2$ have degree 2, $v_3$ has degree 3, and $v_4$
has degree 1 in this subgraph.
In the first round Spoiler pebbles $v_1$.
Suppose that Duplicator responds with a vertex $u_1$ in a connected component $B$ of $H$.
By Lemma~\ref{lem:olariu}, $B$ is either $K_3$-free or a multipartite graph with at least three parts.
In the former case Spoiler wins by pebbling $v_2$ and $v_3$.
In the latter case Spoiler pebbles $v_4$ in the second round.
The distance between $v_1$ and $v_4$ in $G$ is 2.
If Duplicator responds in a connected component of $H$ other than $B$, then
he loses in the next round. Therefore, Duplicator is forced in the second round to pebble a vertex $u_2$
in the same part of $B$ that contains $u_1$.
In this case, Spoiler wins by pebbling the vertex $v_2$.
Indeed, this vertex is adjacent to $v_1$ and not adjacent to $v_4$, while
$u_1$ and $u_2$ have the same adjacency to any other vertex in~$H$.
\end{proof}

\subsection{The path subgraph ($P_4$)}\label{ss:path}

$F=P_4$ is a remarkable example showing that
the parameters $W^*[F]$ and $W[F]$ can have different values.
Specifically, we prove that $W^*[P_4]=3$ and $W[P_4]=4$.

\begin{thm}\label{thm:P_4-WGH}
$W^*[P_4]=\alice{P_4}=3$.  
\end{thm}

The proof of Theorem \ref{thm:P_4-WGH} is based on a well-known characterization
of the class of $P_4$-free graphs.
A graph is called a \emph{cograph} if 
it can be built from copies of the single-vertex graph $K_1$ by using
disjoint unions and complementations.
It is known \cite{CorneilLS81} that a graph is $P_4$-free if and only if it is a cograph.
For the proof of Theorem \ref{thm:P_4-WGH} we need the following definitions, that we borrow from \cite{PikhurkoSV07}.

We call $G$ \emph{complement-connected} if both $G$ and $\compl G$
are connected. An inclusion-maximal complement-connected induced subgraph
of $G$ will be called a \emph{complement-connected component} of $G$
or, for brevity, a \emph{cocomponent} of $G$. Cocomponents have no common 
vertices and their vertex sets form a partition of $V(G)$.
Note that $G$ is a cograph if and only if all cocomponents of $G$ are single-vertex graphs.

The {\em decomposition\/} of $G$, denoted by $\deco G$, is the set of
all connected components of $G$.
Furthermore, given $i\ge0$, we define the
{\em depth $i$ decomposition\/} $\dec iG$ of $G$ by
 $$
 \dec 0G=\deco G\quad \mbox{and}\quad
\dec{i+1}G=\bigcup_{E\in\dec iG}\deco\compl E.
 $$
Note that 
\begin{equation}
  \label{eq:Pi}
\Pi_i=\setdef{V(E)}{E\in\dec iG}  
\end{equation}
 is a partition of $V(G)$, and $\Pi_{i+1}$ refines $\Pi_i$.
Once the partition stabilizes, that is, $\Pi_{i+1}=\Pi_i$,
it coincides with the partition of $G$ into its cocomponents.
The {\em depth $i$ environment\/} of a vertex $v\in V(G)$,
denoted by $\env i{v}$, is the graph $E$ in $\dec iG$ containing~$v$. 

\begin{lem}\label{lem:cographs}
Suppose that a graph $G$ contains an induced copy of $P_4$
and let $U\subseteq V(G)$ be such that $G[U]\cong P_4$.
Consider the 3-pebble Ehrenfeucht-Fra{\"\i}ss{\'e} game on $G$ and another graph $H$.
Let $x,y\in V(G)$ and $x',y'\in V(H)$, and assume that
the pairs $x,x'$ and $y,y'$ were selected by the players in the same rounds.
If $x,y\in U$ and $\env l{x'}\ne\env l{y'}$, then Spoiler has a strategy
allowing him to win in this position playing all the time in $U$
and making no more than $2l+2$ moves.
\end{lem}

\begin{proof}
We use induction on $l$. In the base case of $l=0$,
the vertices $x'$ and $y'$ lie in different connected components of $H$,
while the distance between $x$ and $y$ in $G$ is at most 3.
Therefore, Spoiler is able to win with one extra pebble in 2 moves.

Let $l\ge1$. Suppose that $\env{l-1}{x'}=\env{l-1}{y'}=E$ (for else Spoiler wins by the
induction assumption). Note that $\env l{x'}$ and $\env l{y'}$ are connected
components of $\compl E$. Since $P_4$ is self-complementary, $\compl{G[U]}\cong P_4$.
Therefore, if Duplicator makes further moves only in $V(E)$, Spoiler will win in
at most 2 next moves. Once Duplicator makes one of these moves outside $V(E)$,
this creates a position with two vertices $\tilde x$ and $\tilde y$
pebbled by Spoiler in $U$ such that $\env{l-1}{\tilde x'}\ne\env{l-1}{\tilde y'}$
for the corresponding vertices $\tilde x'$ and $\tilde y'$ pebbled by
Duplicator in $H$. The induction assumption applies.
\end{proof}

\begin{proof}[Proof of Theorem \ref{thm:P_4-WGH}]
Consider graphs $G\sqsupset P_4$ and $H\not\sqsupset P_4$.
Since $H$ is a cograph, $\dec iH$ for some $i$ consists of single-vertex graphs.
By Lemma \ref{lem:cographs}, this readily implies that $W(G,H)\le3$.
Indeed, when Spoiler pebbles two vertices on an induced $P_4$ in $G$,
then whatever Duplicator responds, this creates a position as in Lemma~\ref{lem:cographs}.
Thus, $W^*[P_4]\le3$.

The lower bound $\alice{P_4}>2$ is certified, for example, by~$C_4$.
\end{proof}

\begin{rem}
Whereas $\alice{P_4}\le W^*[P_4]=3$, the upper bound $\alice{P_4}\le3$
follows also from a more direct argument.
The recursive definition of a cograph implies 
that every $P_4$-free graph $H$
contains twins. A pair of twins in $H$ prevents $H\models \ea 3$.
\end{rem}

Theorem \ref{thm:P_4-WGH} implies that the class of graphs containing an induced $P_4$
is definable in the infinitary logic $L^3_{\infty\omega}$. It turns out that this
class is not definable in 3-variable first-order logic.

\begin{thm}\label{thm:P_4-W}
$W[P_4]=4$.  
\end{thm}

Our proof of Theorem \ref{thm:P_4-W} is based on Lemma \ref{lem:WDGH}.
It suffices to exhibit a sequence of graph pairs $G_i,H_i$ such that $G_i$
contains an induced copy of $P_4$, $H_i$ does not, and $D^3(G_i,H_i)\to\infty$ as $i$ increases.

Given a graph $X$, we define its $i$th power $X^i$ by $X^1=X$ and $X^{i+1}=\compl{X^i+X^i}$.
Now, let $H_i=(K_1)^i$. This is a cograph and, therefore, $P_4\not\sqsubset H_i$ (which is also easy to see
directly, using induction and the fact that $P_4$ is self-complementary).

In order to construct $G_i$, we use the lexicographic product of graphs; see Section \ref{s:prel}.
Fix a graph $A$ satisfying the 3rd extension axiom and containing $P_4$ as an
induced subgraph (a large enough random graph has both of these properties
with high probability). Now, let $G_i=H_i\cdot(A\cdot H_i)$.
Obviously, $P_4\sqsubset G_i$. Theorem \ref{thm:P_4-W} immediately follows by Lemma \ref{lem:WDGH}
from the following estimate.

\begin{lem}\label{lem:P_4-D^3}
$D^3(G_{4m+3},H_{4m+3}) \ge m+2$.
\end{lem}

The proof of Lemma \ref{lem:P_4-D^3} takes the rest of this subsection.

\subsubsection*{Proof of Lemma \protect\ref{lem:P_4-D^3}}

We have to describe a strategy of Duplicator in the 3-pebble game on $G=G_{4m+3}$ and $H=H_{4m+3}$
allowing her to resist during $m+1$ rounds.
This will be convenient to do in terms of a nonstandard metric on $H$, which we
introduce now.

Let $t=4m+2$; thus, $H=H_{t+1}$. 
Consider the partitions of $\Pi_0,\ldots,\Pi_t$ of $V(H)$ defined by \refeq{Pi}.
Note that $\Pi_0$ is the trivial partition
with the single partition element $V(H)$, and $\Pi_t$ is the complete
partition of $V(H)$ into singletons.
Let $x$ and $y$ be vertices of $H$. Define $\bar d(x,y)$ to be the maximum $k$ such that
$x$ and $y$ belong to the same element of $\Pi_k$ or, equivalently, $\env kx=\env ky$.
Furthermore, we set $d(x,y)=t-\bar d(x,y)$. Note that $0\le d(x,y)\le t$.
As easily seen, $d(x,y)=0$ iff $x=y$, and $d(x,y)=1$ iff $x$ and $y$ are twins.

Note that adjacency between two vertices $x$ and $y$ is completely determined by $d(x,y)$,
and it changes when $d(x,y)$ increases or decreases by~$1$.

\begin{claim}\label{cl:auto}
  \begin{enumerate}[1.]
  \item 
$H$ is vertex-transitive.
\item 
Every automorphism $\alpha$ of $H$ preserves the function $d$, i.e., $d(\alpha(x),\alpha(y))=d(x,y)$ for any $x,y\in V(H)$.
\item 
Let $x,y,x',y'\in V(H)$. There exists an automorphism $\alpha$ of $H$ such that $\alpha(x)=x'$ and $\alpha(y)=y'$
if and only if $d(x,y)=d(x',y')$.
  \end{enumerate}
\end{claim}

\begin{subproof}
The claim holds true for every $H=H_i$ and follows by a simple induction on $i$.
Part 1 follows from an observation that $H_i=\compl{H_{i-1}+H_{i-1}}$ has an automorphism
transposing the two connected components of~$\compl{H_i}$.

Let $d_i$ denote the distance function on $H_i$ defined similarly to the
distance function $d$ on $H$. Thus, $d=d_{4m+3}$ and $d_i$ takes on the values
from $0$ to $i-1$. Moreover, the restriction of $d_i$ to a connected component
of $\compl{H_i}$ coincides with $d_{i-1}$ on this component. The parameterized
version of Part 2 claims that $d_i$ is preserved by every automorphism $\alpha$
of $H_i$. Note that $\alpha$ either maps each of the two connected components
of $\compl{H_i}$ onto itself or swaps them. If $x$ and $y$ are in different 
connected components of $\compl{H_i}$, the same holds true for $\alpha(x)$ and 
$\alpha(y)$. In this case, $d_i(\alpha(x),\alpha(y))=i-1=d_i(x,y)$.
If $x$ and $y$ are in the same connected component of $\compl{H_i}$,
we have $d_i(\alpha(x),\alpha(y))=d_{i-1}(\alpha(x),\alpha(y))=d_{i-1}(x,y)=d_i(x,y)$ 
by the induction assumption.

In the parameterized version of Part 3, we have to prove that an automorphism of $H_i$
taking $x$ to $x'$ and $y$ to $y'$ exists if and only if $d_i(x,y)=d_i(x',y')$.
Consider two cases.
First, suppose that $x$ and $y$ are in different $H_{i-1}$-components of $\compl{H_i}$,
that is, $d_i(x,y)=i-1$.
If $d_i(x,y)=d_i(x',y')$, the vertices $x'$ and $y'$ are also in different $H_{i-1}$-components,
and we use an automorphism transposing these components and the vertex-transitivity
of $H_{i-1}$ established in Part 1. If $d_i(x,y)\ne d_i(x',y')$, then such an automorphism
does not exist by Part 2.
The case that $x'$ and $y'$ are in different $H_{i-1}$-components of $\compl{H_i}$ is symmetric.
There remain the cases that all four $x,y,x',y'$ are in the same $H_{i-1}$-component of $\compl{H_i}$
or that $x,y$ are in one $H_{i-1}$-component of $\compl{H_i}$ and $x',y'$ are in the other component.
They are covered by the induction assumption.
\end{subproof}

\begin{claim}\label{cl:ddd}
$d(x,y)=d(x,z)=d(y,z)$ is impossible for any pairwise distinct $x,y,z\in V(H)$.
\end{claim}

\begin{subproof}
These equalities would imply that $\bar d(x,y)=\bar d(x,z)=\bar d(y,z)=k$
for some $k<t$.
This would mean that $\env{k+1}x$, $\env{k+1}y$, and $\env{k+1}z$ are
pairwise disjoint while $\env{k}x=\env{k}y=\env{k}z$. This is a contradiction
because a depth-$k$ environment has only two depth-$(k+1)$ subenvironments.
\end{subproof}

\begin{claim}\label{cl:d<d}
If $d(x,y)<d(x,z)$, then $d(y,z)=d(x,z)$.
\end{claim}

\begin{subproof}
Denote $k=\bar d(x,y)$. Then $\env{k}y=\env{k}x$ has an empty intersection with $\env{k}z$.
It follows that $\env{m}z=\env{m}y$ iff $\env{m}z=\env{m}x$.
\end{subproof}

Claims \ref{cl:ddd} and \ref{cl:d<d} imply that, with respect to the distance function $d$,
any three vertices of $H$ form an isosceles triangle where the base is shorter than the two legs.
This implies that $d$ is a graph metric because
\begin{equation}
  \label{eq:d-max}
d(x,y)\le\max(d(x,z),d(z,y))  
\end{equation}
and, hence, $d(x,y)\le d(x,z)+d(z,y)$.

Assume that, in the course of the 3-pebble game on $G$ and $H$, the players pebble
vertices $x$ and $y$ in $G$ and the corresponding vertices $x'$ and $y'$ in $H$
so that the position is not an immediate loss for Duplicator. 
Note that there are two configurations that are potentially dangerous
for Duplicator. If $d(x',y')=1$, then Duplicator loses whenever
Spoiler pebbles a vertex $z$ in $G$ adjacent to exactly one of $x$ and $y$.
If $d(x',y')=t$, then Duplicator loses whenever
Spoiler pebbles a vertex $z$ in $G$ adjacent neither to $x$ nor to~$y$.
Any other configuration is non-losing for Duplicator if Spoiler moves in $G$,
since any other pair $x',y'$ has all possible extensions: a common neighbor, a common non-neighbor, 
a vertex adjacent only to $x'$ and a vertex adjacent only to $y'$. 
In particular, we will use the following fact.

\begin{claim}\label{cl:nonlosing}
Suppose that $2<d(x',y')<t-2$. Then, whatever vertex $z$ is pebbled by Spoiler in $G$,
Duplicator has a non-losing move in $H$. Moreover, there is
a vertex $z'\in V(H)$ with any desired adjacency to $x'$ and to $y'$ such that
\begin{equation}
  \label{eq:dd2}
|d(z',x')-d(x',y')|\le2\text{ and }|d(z',y')-d(x',y')|\le2.  
\end{equation}
\end{claim}

\begin{subproof}
Denote $k=\bar d(x',y')$.
Assume that $x'\sim y'$ and, hence $x\sim y$ (the case that $x'\not\sim y'$ is symmetric).
If $z\sim x$ and $z\not\sim y$, then Duplicator pebbles a vertex $z'$ in $\env{k+1}{y'}$.
This will ensure that $z'\sim x'$. More precisely,
Duplicator chooses $z'\in\env{k+1}{y'}\setminus\env{k+2}{y'}$ to ensure also that $z'\not\sim y'$.
Note that $\bar d(z',x')=k$ and $\bar d(z',y')=k+1$.
This implies \refeq{dd2}. The case that $z\not\sim x$ and $z\sim y$ is symmetric.

If $z\not\sim x$ and $z\not\sim y$, then Duplicator pebbles a vertex $z'$ in $\env{k-1}{x'}\setminus\env{k}{x'}$
(these environments are the same for $y'$). This ensures that $z'\not\sim x'$ and $z'\not\sim y'$.
Moreover, $\bar d(z',x')=\bar d(z',y')=k-1$.
If $z\sim x$ and $z\sim y$, then Duplicator pebbles a vertex $z'$ in $\env{k-2}{x'}\setminus\env{k-1}{x'}$.
Note that $\bar d(z',x')=\bar d(z',y')=k-2$ in this case.
\end{subproof}

We now make a few useful remarks on the structure of~$G$.

\begin{claim}\label{cl:coconn}
For every $i$, the graph $A\cdot H_i$ is complement-connected.
\end{claim}

\begin{subproof}
By the 3-extension property, both $A$ and $\compl A$ have diameter 2 and, hence, are connected.
Note that $\compl{A\cdot H_i}=\compl A\cdot\compl{H_i}$. The claim follows from the fact that
the lexicographic product is connected whenever the first factor is connected and has more than one vertex.
\end{subproof}

Note that $G_i=H_i\cdot(A\cdot H_i)$ is obtained from $H_i$ by substituting 
each vertex for a copy of $A\cdot H_i$. It follows by Claim \ref{cl:coconn} that 
all cocomponents of $G_i$ are isomorphic to $A\cdot H_i$ or its complement,
depending on the parity of $i$.
Furthermore, define an operation on graphs $f$ by $f(X)=\compl{X+X}$.
Note that $X^i=f^{i-1}(X)$ and $H_i=f^{i-1}(K_1)$, where $f^k$ denotes the $k$-fold composition of $f$.
Using the equality
$$
f(X)=\compl{(K_1+K_1)\cdot X}=K_2\cdot\compl{X}=f(K_1)\cdot\compl{X},
$$
an easy inductive argument yields
$$
f^j(X)=
\begin{cases}
f^j(K_1)\cdot\compl{X} & \text{if $j$ is odd},\\
f^j(K_1)\cdot X & \text{if $j$ is even}.
\end{cases}
$$
Suppose now that $i$ is odd (like in our case of $i=4m+3$). Then
$$
G_i=H_i\cdot(A\cdot H_i)=f^{i-1}(K_1)\cdot(A\cdot H_i)=f^{i-1}(A\cdot H_i)=(A\cdot H_i)^i.
$$
This reveals that $G_i=(A\cdot H_i)^i$ has the same recursive structure
as $H_i=(K_1)^i$, where each cocomponent of $G_i$ is isomorphic to $A\cdot H_i$.

Being isomorphic to $A\cdot H$, each cocomponent of $G$ consists, therefore,
of $v_A$ copies of $H$ (and edges between them), where $v_A$ is the number of vertices in $A$.
On each of the $v_A$ copies of $H$ we consider the metric $d$ introduced above.
We also consider this metric on the first factor of $H$ in $G=H\cdot(A\cdot H)$,
using the different notation $D$ for it. Moreover, we extend the distance functions $D$ and $d$
to the entire set $V(G)^2$ as follows. Recall that a vertex of $G$ is a triple $(h,a,h')$
where $h,h'\in V(H)$ and $a\in A$. Then
$$
D((h_1,a_1,h'_1),(h_2,a_2,h'_2))=d(h_1,h_2)
$$
and
$$
d((h_1,a_1,h'_1),(h_2,a_2,h'_2))=
\begin{cases}
d(h'_1,h'_2)&\text{if }(h_1,a_1)=(h_2,a_2),\\
\infty&\text{otherwise}.
\end{cases}
$$
Similarly to the dangerous configurations for Duplicator that are mentioned above,
note that Duplicator can lose in the next round if two vertices $x$ and $y$
are pebbled in $G$ such that $d(x,y)=1$ or $D(x,y)=t$.
On the other hand, a configuration $x,y\in V(G)$ with $d(x,y)=t$
is not dangerous because if Spoiler pebbles $z'\in V(H)$ with the same
adjacency to $x'$ and $y'$, then Duplicator can respond with $z$ 
in the same $(A\cdot H)$-cocomponent of $G$ but in a different copy of $H$
(not containing $x$ and $y$). She can ensure adjacency or non-adjacency to
both $x$ and $y$ because $A$ has the 3-extension property.
A configuration with $D(x,y)=1$ also does not pose any threat to Duplicator
because she can find a vertex $z$ adjacent to exactly one of $x$ and $y$
either in the $H$-part of a cocomponent of $G$ containing $x$ or in the
$H$-part of a cocomponent of $G$ containing $y$ (that is, either $d(z,x)<\infty$
or $d(z,y)<\infty$).

Speaking of a \emph{2-pebble configuration $x/x'$, $y/y'$}, where $x,y\in V(G)$
and $x',y'\in V(H)$, we mean that the vertices $x$ and $y$ are pebbled in $G$
and the corresponding pebbles occupy the vertices $x'$ and $y'$ in $H$.
Let $0\le s\le 2m$.
We call a configuration $x/x'$, $y/y'$ \emph{$s$-safe} (for Duplicator)
if the following three conditions are true:
\begin{itemize}
\item 
$x\sim y$ iff $x'\sim y'$;
\item 
$d(x,y)=d(x',y')$ or both the distances are sufficiently large: $d(x,y)>s$ and $d(x',y')>s$.
\item 
$D(x,y)=d(x',y')$ or both the distances are sufficiently small: $D(x,y)<t-s$ and $d(x',y')<t-s$.
\end{itemize}
A 3-pebble configuration $x/x'$, $y/y'$, $z/z'$ is \emph{$s$-safe}
if every 2-pebble subconfiguration of it is $s$-safe.

We are now ready to describe Duplicator's strategy.

\smallskip

\textit{1st round.}
Duplicator's response is arbitrary. If she moves in $H$, recall that $H$ is vertex-transitive.
If she moves in $G$, there can be non-isomorphic choices but, as we will see, a particular choice does
not influence the game.\footnote{In fact, we can make $G$ vertex-transitive too
by choosing $A$ to be the Paley graph of order~$13$.}
Suppose that the players have pebbled
vertices $x\in V(G)$ and $x'\in V(H)$.

\smallskip

\textit{2nd round.}
The vertices pebbled in this round will be denoted by $y$ and $y'$.
Duplicator ensures a $(2m-1)$-safe configuration $x/x'$, $y/y'$ as follows:
\begin{itemize}
\item 
If $d(x,y)\le2m-1$ or $d(x',y')\le2m-1$ (depending on the graph Spoiler moves in),
then Duplicator responds so that $d(x,y)=d(x',y')$. Note that this guarantees that
$x\sim y$ iff $x'\sim y'$.
\item 
If $D(x,y)\ge2m+3$ or $d(x',y')\ge2m+3$ (depending on the graph Spoiler moves in),
then Duplicator responds so that $D(x,y)=d(x',y')$. Again, this guarantees that
$x\sim y$ iff $x'\sim y'$.
\item 
Otherwise, Duplicator takes care that none of the above conditions is true,
that is, both $d(x,y)$ and $d(x',y')$ are at least $2m$
and both $D(x,y)$ and $d(x',y')$ are at most $2m+2$.
If moving in $H$, she achieves this by pebbling a vertex $y'$ such that
$d(x',y')=2m$ or $d(x',y')=2m+1$, depending on whether or not $x$ and $y$ are adjacent.
If moving in $G$, Duplicator can achieve the desired adjacency between $x$ and $y$,
for example, by pebbling $y$ such that $d(x,y)=\infty$ and $D(x,y)=0$
(due to the 2-extension property of~$A$). 
\end{itemize}

The core of Duplicator's strategy in further rounds is stated in the following claim.

\begin{claim}\label{cl:ssafe}
Let $2\le s\le 2m-1$. Suppose that $x/x'$, $y/y'$ is an $s$-safe 2-pebble configuration.
Then, whatever Spoiler does in the next round, Duplicator can respond so that the
resulting configuration $x/x'$, $y/y'$, $z/z'$ is $(s-2)$-safe.
\end{claim}

Claim \ref{cl:ssafe} readily implies that Duplicator survives not only
in the first and the second rounds but also in at least $m-1$ subsequent rounds,
which yields the bound in Lemma \ref{lem:P_4-D^3}.
It remains to prove this claim.

\smallskip

\begin{subproof}
\Case A{$d(x',y')\le s$}.
Since the configuration $x/x'$, $y/y'$ is $s$-safe, we have $d(x,y)=d(x',y')$
and, hence, $x\sim y$ iff $x'\sim y'$.
We split our analysis into several subcases depending on
the mutual position of the vertices $z$ and $x$ (or $z'$ and $x'$).
The mutual position of the vertices $z$ and $y$ (or $z'$ and $y'$)
will be determined in each of the subcases
because the vertices $x$ and $y$ (and  $x'$ and $y'$) are close to each other in the metric~$d$.

\Subcase{A-1}{Spoiler pebbles $z\in V(G)$ such that $d(z,x)\le s$.}
It follows by \refeq{d-max} that $d(z,y)\le s$.
Let $H'$ be the $H$-part of a cocomponent of $G$ containing the vertices $x$ and $y$.
By Claim \ref{cl:auto}, there is an isomorphism $\alpha$ from $H'$ to $H$ such that
$\alpha(x)=x'$ and $\alpha(y)=y'$. Duplicator pebbles the vertex $z'=\alpha(z)$.
The adjacencies are respected automatically. Since 
\begin{eqnarray*}
d(z',x')=d(z,x)&\le&s<t-s\text{\ \ and}\\
d(z',y')=d(z,y)&\le&s<t-s,  
\end{eqnarray*}
the configuration $x/x'$, $y/y'$, $z/z'$ is $s$-safe and, hence, $(s-2)$-safe.

\Subcase{A-1$\,{}'$}{Spoiler pebbles $z'\in V(H)$ such that $d(z',x')\le s$.}
Duplicator responds in $G$ mirroring Subcase~A-1.

\Subcase{A-2}{Spoiler pebbles $z\in V(G)$ such that $D(z,x)\ge t-s$.}
Since $x$ and $y$ are in the same cocomponent and $z$ is in a different cocomponent of $G$,
we have $D(z,x)=D(z,y)$, and $z\sim x$ iff $z\sim y$.
Duplicator pebbles $z'\in V(H)$ such that $d(z',x')=D(z,x)$. 
This implies that $z'\sim x'$ iff $z\sim x$. Since
$d(z',x')\ge t-s>s\ge d(x',y')$, Claim \ref{cl:d<d} implies that $d(z',y')=d(z',x')=D(z,x)=D(z,y)$.
As a consequence, $z'\sim y'$ iff $z\sim y$.
It follows that the resulting configuration is $s$-safe.

\Subcase{A-2$\,{}'$}{Spoiler pebbles $z'\in V(H)$ such that $d(z',x')\ge t-s$.}
Duplicator responds in $G$ mirroring Subcase~A-2.

\Subcase{A-3.1}{Spoiler pebbles $z\in V(G)$ such that $s<d(z,x)<\infty$.}
Since $d(z,y)<\infty$ and $d(z,x)>d(x,y)$, we have $d(z,y)=d(z,x)$ by Claim \ref{cl:d<d},
and $z\sim y$ iff $z\sim x$.
Duplicator responds with $z'\in V(H)$ such that $d(z',x')=2m$ or $d(z',x')=2m+1$.
More specifically, she
chooses the option ensuring that $z'\sim x'$ iff $z\sim x$. In either case, $d(z',x')>d(x',y')$ and, hence,
$d(z',y')=d(z',x')$. Therefore, 
$$
z'\sim y'\iff z'\sim x'\iff z\sim x\iff z\sim y,
$$
and the new configuration remains $s$-safe.

\Subcase{A-3.2}{Spoiler pebbles $z\in V(G)$ such that $0<D(z,x)<t-s$.}
Note that $D(z,y)=D(z,x)$. Duplicator plays like in Subcase~A-3.1.

\Subcase{A-3.3}{Spoiler pebbles $z\in V(G)$ such that $d(z,x)=\infty$ and $D(z,x)=0$.}
Note that $d(z,y)=\infty$ and $D(z,y)=0$, which implies that $z\sim y$ iff $z\sim x$. 
Duplicator plays like in Subcase~A-3.1.

\Subcase{A-3$\,{}'$}{Spoiler pebbles $z'\in V(H)$ such that $s<d(z',x')<t-s$.}
Since $d(z',x')>s\ge d(x',y')$, we have $d(z',y')=d(z',x')$ and, hence,
$z'\sim y'$ iff $z'\sim x'$.
Duplicator can, for example, mirror Subcase~A-3.3 by pebbling $z\in V(G)$
such that $d(z,x)=\infty$, $D(z,x)=0$, and $z\sim x$ iff $z'\sim x'$.
The last equivalence can be achieved due to the 2-extension property of the graph~$A$.
The equivalence $z\sim y\iff z'\sim y'$ will be true automatically
(by Claim \ref{cl:d<d} and the structure of~$G$).

\Case B{$d(x',y')\ge t-s$}. Since the configuration $x/x'$, $y/y'$ is $s$-safe, $D(x,y)=d(x',y')$ 
and $x\sim y$ iff $x'\sim y'$ in this case.
We split analysis of this case into several subcases depending on
the mutual position of the vertices $z$ and $x$ (or $z'$ and $x'$).
There are also symmetric subcases depending on
the mutual position of $z$ and $y$ (or $z'$ and $y'$),
that are omitted below.

\Subcase{B-1}{Spoiler pebbles $z\in V(G)$ such that $d(z,x)\le s-2$.}
Note that $D(z,y)=D(x,y)$. Duplicator pebbles $z'$ in $H$ such that
$d(z',x')=d(z,x)$. Since $d(z',x')\le s-2<t-s\le d(x',y')$, we have
$d(z',y')=d(x',y')=D(x,y)=D(z,y)$, which implies the right adjacencies.
As easily seen, the new configuration is $(s-2)$-safe.

\Subcase{B-1$\,{}'$}{Spoiler pebbles $z'\in V(H)$ such that $d(z',x')\le s-2$.}
This case is, in a sense, mirror-symmetric to Subcase~B-1. 
Note that $d(z',y')=d(x',y')$. Duplicator pebbles $z$ in $G$ such that
$d(z,x)=d(z',x')$, which implies that $D(z,y)=D(x,y)$.

\Subcase{B-2}{Spoiler pebbles $z\in V(G)$ such that $D(z,x)>s-2$ and $D(z,y)>s-2$.}
Duplicator plays according to an isomorphism from the first $H$-factor in $G=H\cdot(A\cdot H)$
to the graph $H$ taking $x$ to $x'$ and $y$ to $y'$, which exists by Claim~\ref{cl:auto}.

\Subcase{B-2$\,{}'$}{Spoiler pebbles $z'\in V(H)$ such that $d(z',x')>s-2$ and $d(z',y')>s-2$.}
Duplicator mirrors her strategy from Subcase~B-2.

\Subcase{B-3}{Spoiler plays in $G$ but differently from Subcases B-1 and B-2. That is,
he pebbles $z\in V(G)$ such that $0<D(z,x)\le s-2$ or such that 
$D(z,x)=0$ and $s-2<d(z,x)\le\infty$.}
Since $D(x,y)\ge t-s>s-2\ge D(z,x)$, we have
$D(z,y)=D(y,x)$ and $z\sim y$ iff $y\sim x$. Duplicator pebbles $z'$ in $H$ such that
$d(z',x')=2m$ or $d(z',x')=2m+1$;
more specifically, she chooses the option ensuring that $z'\sim x'$ iff $z\sim x$.
Since $d(z',x')<t-s\le d(x',y')$, we have $d(x',y')=d(z',y')$ and, therefore,
$$
z'\sim y'\iff x'\sim y'\iff x\sim y\iff z\sim y,
$$
yielding the right adjacencies.
It is straightforward to check that the configuration $x/x'$, $y/y'$, $z/z'$ is $(s-2)$-safe.

\Case C{$s<d(x',y')<t-s$.}
Note that in this case we have $d(x,y)>s$ and $D(x,y)<t-s$.

\Subcase{C-1}{$s<d(x,y)<\infty$.}

\Subcase{C-1.1}{Spoiler pebbles $z\in V(G)$ such that $d(z,x)\le s-2$ or $d(z,y)\le s-2$.}
Duplicator pebbles $z'\in V(H)$ with $d(z',x')=d(z,x)$ or $d(z',y')=d(z,y)$
respectively. 

\Subcase{C-1.1$\,{}'$}{Spoiler pebbles $z'\in V(H)$ such that $d(z',x')\le s-2$ or $d(z',y')\le s-2$.}
This is the counterpart of Subcase~C-1.1.

\Subcase{C-1.2}{Spoiler pebbles $z\in V(G)$ such that $D(z,x)=D(z,y)\ge t-s+2$.}
Duplicator responds with $z'\in V(H)$ such that $d(z',x')=d(z',y')=D(z,x)=D(z,y)$.

\Subcase{C-1.2$\,{}'$}{Spoiler pebbles $z'\in V(H)$ such that $d(z',x')=d(z',y')\ge t-s+2$.}
This is the counterpart of Subcase~C-1.2.

\Subcase{C-1.3}{Spoiler pebbles $z\in V(G)$ such that $s-2<d(z,x)<\infty$ and $s-2<d(z,y)<\infty$
or he pebbles $z\in V(G)$ 
such that 
$d(z,x)=d(z,y)=\infty$ and
$D(z,x)=D(z,y)<t-s+2$.}
Duplicator succeeds by Claim~\ref{cl:nonlosing}.

\Subcase{C-1.3$\,{}'$}{Spoiler pebbles $z'\in V(H)$ such that $s-2<d(z',x')<t-s+2$ and 
$s-2<d(z',y')<t-s+2$.}
If $z'$ is adjacent to exactly one of $x'$ and $y'$, then
Duplicator has a successful move $z$ 
such that $d(z,x)<\infty$ (and hence $d(z,y)<\infty$);
see the proof of Claim~\ref{cl:nonlosing} that ensures that $d(z,x)\ge d(x,y)-1>s-2$
and $d(z,y)\ge d(x,y)-1>s-2$. If $z'$ is adjacent to both or none of $x'$ and $y'$, then
Claim~\ref{cl:nonlosing} is non-applicable if $d(x,y)\ge t-2$. In this case,
Duplicator has a successful choice of $z$ in the cocomponent of $G$ containing $x$ and $y$
(due to the 3-extension property of the graph $A$).

\Subcase{C-2}{$d(x,y)=\infty$ and $D(x,y)=0$.}

\Subcase{C-2.1}{Spoiler pebbles $z\in V(G)$ such that $d(z,x)\le s-2$ or $d(z,y)\le s-2$.}
Duplicator pebbles $z'\in V(H)$ with $d(z',x')=d(z,x)$ or $d(z',y')=d(z,y)$
respectively.

\Subcase{C-2.1$\,{}'$}{Spoiler pebbles $z'\in V(H)$ such that $d(z',x')=d(z,x)$ or $d(z',y')=d(z,y)$.}
This is the counterpart of Subcase~C-2.1.

\Subcase{C-2.2}{Spoiler pebbles $z\in V(G)$ such that $D(z,x)=D(z,y)\ge t-s+2$.}
Duplicator responds with $z'\in V(H)$ such that $d(z',x')=d(z',y')=D(z,x)=D(z,y)$.

\Subcase{C-2.2$\,{}'$}{Spoiler pebbles $z'\in V(H)$ such that $d(z',x')=d(z',y')\ge t-s+2$.}
This is the counterpart of Subcase~C-2.2.

\Subcase{C-2.3}{Spoiler plays in $G$ but differently from Subcases C-2.1 and C-2.2.}
Duplicator has an appropriate move in $H$ by Claim~\ref{cl:nonlosing}.

\Subcase{C-2.3$\,{}'$}{Spoiler pebbles $z'\in V(H)$ such that $s-2<d(z',x')<t-s+2$ and 
$s-2<d(z',y')<t-s+2$.}
Duplicator pebbles a vertex $z$ in the cocomponent of $G$ containing $x$ and $y$
but in an $H$-part of this cocomponent different from the $H$-parts containing $x$ and $y$.
She is able to obey the adjacency relation due to the 3-extension property of the graph~$A$.

\Subcase{C-3}{$0<D(x,y)<t-s$.}

Subcases C-3.1, C-3.1${}'$, C-3.2, C-3.2$\,{}'$, and C-3.3 are as in Subcase~C-2.

\Subcase{C-3.3$\,{}'$}{Spoiler pebbles $z'\in V(H)$ such that $s-2<d(z',x')<t-s+2$ and 
$s-2<d(z',y')<t-s+2$.}

If $d(z',x')=d(z',y')>d(x',y')$, Duplicator pebbles a vertex $z$ such that
$D(z,x)=D(z,y)$ is equal either to $D(x,y)+1$ or to $D(x,y)+2$ so that
$z$ has the required adjacency to $x$ and $y$. Note that $D(z,x)=D(z,y)<t-s+2$,
which implies that the configuration is $(s-2)$-safe.

If $d(z',x')<d(z',y')=d(x',y')$, Duplicator pebbles a vertex $z$ such that
$d(z,x)=2m$ or $d(z,x)=2m+1$. 
The configuration is $(s-2)$-safe as $D(z,x)=0$ and $D(z,y)=D(x,y)$.

The case when $z'$ is closer to $y'$ than to $x'$ in the metric $d$ is symmetric.
\end{subproof}

\subsection{The claw ($K_{1,3}$) and the diamond ($K_4\setminus e$) subgraphs}

A \emph{strongly regular graph} with parameters
$(n,k,\lambda,\mu)$ is a regular graph with $n$ vertices of degree $k$
such that every two adjacent vertices have $\lambda$ common neighbors
and every two non-adjacent vertices have $\mu$ common neighbors.
The simplest examples are 
 $sK_t$ (the vertex-disjoint union of $s$ copies of the complete graph $K_t$)
and their complements (complete $s$-partite graphs with each vertex class
of size $t$). We call such strongly regular graphs \emph{trivial}.
A strongly regular graph is non-trivial exactly if $0<\mu<k<n-1$.

An example of a non-trivial strongly regular graph, that will be useful for us below,
is the \emph{$m\times m$-rook graph}. The vertex set of this graph is
$\setdef{(a,b)}{1\le a,b\le m}$, and two vertices $(a_1,b_1)$ and $(a_2,b_2)$
are adjacent if $a_1=a_2$ or $b_1=b_2$. In other words, each vertex represents a square
of the $m\times m$ chess board, and two squares are adjacent if one is reachable from the
other by a move of the rook. The $m\times m$-rook graph is strongly regular with
parameters $(m^2,2m-2,m-2,2)$.

The condition $\lambda=0$ means that a strongly regular graph is $K_3$-free.
Every complete bipartite graph $K_{n,n}=\compl{2K_n}$ has 
this property and seven other triangle-free non-trivial graphs are known.
It is open whether there is yet another such graph~\cite{Godsil95}.

Suppose that $H$ is a non-trivial non-triangle-free strongly regular graph
with parameters $(n,k,\lambda,\mu)$. Thus, $\mu<k$ and it is also not hard to see that 
$\lambda<k-1$ (otherwise every two adjacent vertices were twins and, by connectedness,
every two vertices would be adjacent twins, implying that $H$ is complete).
These two inequalities readily imply that $H$ satisfies the 3rd extension axiom.

\begin{thm}\label{thm:K_13-W-ind}
  $W[K_{1,3}]=W^*[K_{1,3}]=\alice{K_{1,3}}=4$ and $W[K_4\setminus e]=W^*[K_4\setminus e]=\alice{K_4\setminus e}=4$.
\end{thm}

\begin{proof}
We have to show that $\alice{K_{1,3}}>3$ and $\alice{K_4\setminus e}>3$.
By the discussion above, it suffices to exhibit a non-trivial non-triangle-free strongly regular graph
that neither contains the claw graph nor the diamond graph as induced subgraphs.
The $3\times3$-rook graph suits these needs.
\end{proof}

\subsection{The cycle subgraph ($C_4$)}

Let $G$ be a connected graph. Given $u,v\in V(G)$, let $f^G_{i,j}(u,v)$ denote
the number of vertices at distance $i$ from $u$ and at distance $j$ from $v$.
The graph $G$ is called \emph{distance-regular} if the number $f^G_{i,j}(u,v)=f^G_{i,j}(d)$
depends only on $i$, $j$, and the distance $d=d(u,v)$ between $u$ and $v$.
Note that such a graph is regular of degree $f^G_{1,1}(0)$.
We call two distance-regular graphs $G$ and $H$ \emph{similar} if 
\begin{equation}
  \label{eq:dist-reg}
f^G_{i,j}(d)=0\iff f^H_{i,j}(d)=0.  
\end{equation}

\begin{lem}\label{lem:dist-reg}
If $G$ and $H$ are similar distance-regular graphs, then $W(G,H)>3$.
\end{lem}

\begin{proof}
We show a strategy allowing Duplicator to win the 3-pebble game on $G$ and $H$.
In the first round she responds Spoiler's move arbitrarily.
Let $x$ and $x'$ be the vertices pebbled in $G$ and $H$ respectively.
Suppose that in the second round Spoiler pebbles a vertex $y$ in $G$
(the case that Spoiler plays in $H$ is similar).
Duplicator responds with a vertex $y'$ in $H$ such that $d(x',y')=d(x,y)$,
which guarantees that she does not lose in this round. Such a choice of $y'$
is possible because \refeq{dist-reg} implies that $f^H_{i,i}(0)>0$ for $i=d(x,y)$.

For any subsequent round, assume that $x,y\in V(G)$ and $x',y'\in V(H)$
are occupied by two pairs of pebbles and that $d(x',y')=d(x,y)$.
Suppose that Spoiler puts the third pebble
on a vertex $z$ in $G$ (the case that Spoiler plays in $H$ is similar).
It is enough to notice that Duplicator can pebble a vertex $z'$ in $H$ such that
$d(z',x')=d(z,x)$ and $d(z',y')=d(z,y)$. Such a vertex exists because \refeq{dist-reg}
implies that $f^H_{i,j}(d)>0$ for $i=d(z,x)$, $j=d(z,y)$, and $d=d(x,y)$.
\end{proof}

\begin{thm}\label{thm:C_4}
 $W[C_4]=W^*[C_4]=4$. 
\end{thm}

\begin{proof}
By Lemma \ref{lem:dist-reg}, it suffices to exhibit similar
distance-regular graphs $G$ and $H$ such that 
$G$ contains an induced copy of $C_4$ and $H$ does not.
We can take $G$ to be the cubic graph (the skeleton of the 3-dimensional cube)
and $H=C_6$.
\end{proof}

\section{Lower bounds over highly connected graphs}\label{s:highly}

As we discussed in Section \ref{s:intro},
in the case of a connected pattern graph $F$
it is algorithmically motivated to consider the parameter $W_\kappa[F]$,
which is the minimum variable width of a sentence defining 
the graph class $\indsubgr F$ correctly only over graphs of \emph{sufficiently large} connectedness.
More precisely, $W_\kappa[F]$ is equal to the minimum $k$
for which there is a $k$-variable sentence $\Phi$ and a number $s$ such that 
$G\models\Phi$ iff $F\sqsubset G$ for all $s$-connected graphs $G$.
Recall that any lower bound for $W_\kappa[F]$ is also a lower bound for $W_\tw[F]$,
which rules out some efficient approaches to \textsc{Induced Subgraph Isomorphism}
based on Courcelle's theorem.

Moreover, we define 
$$
W^*_\kappa[F]=\textstyle\min_s\max\setdef{W(G,H)}{F\sqsubset G,\,F\not\sqsubset H,\text{ and both $G$ and $H$ are $s$-connected}}.
$$
This parameter is an analog of $W_\kappa[F]$ for the infinitary logic, and we have 
$$
W^*_\kappa[F]\le W_\kappa[F]\le W[F]\text{ and }W^*_\kappa[F]\le W^*[F]\le W[F].
$$

The \emph{join} of graphs $A$ and $B$ is denoted by $A*B$.
Recall that this is the graph obtained from the disjoint union of $A$ and $B$
by adding all possible edges between a vertex of $A$ and a vertex of~$B$.

\begin{lem}\label{lem:join}
$W(A*B,A'*B)\ge W(A,A')$.  
\end{lem}

\begin{proof}
In the game on $A*B$ and $A'*B$, Duplicator can use her strategy for the game on $A$ and $A'$.
Each time that Spoiler moves in the $B$ part of one graph, Duplicator just mirrors his move in 
the $B$ part of the other graph.
\end{proof}

\begin{lem}\label{lem:padding}
Let $F_0$ be obtained from $F$ by removing all universal vertices from this graph. Then
$W^*_\kappa[F]\ge W^*[F_0]$.
\end{lem}

\begin{proof}
Let $m$ denote the number of universal vertices in $F$.
Suppose that $W^*[F_0]=W(G,H)$, where $G$ contains an induced copy of $F_0$
and $H$ does not. Let $G'$ be obtained from $G$ by adding new $s>m$ universal vertices,
and let $H'$ be defined similarly, i.e, $G'=G*K_s$ and $H'=H*K_s$. Note that 
$G'$ contains an induced copy of $F$ and $H'$ still does not contain even an induced
copy of $F_0$ (no new vertex of $H'$ can appear in an induced copy of $F_0$
because it would be universal there). We have $W(G',H')\ge W(G,H)$
by Lemma \ref{lem:join}. This proves the claim because
$G'$ and $H'$ are $s$-connected and $s$ can be chosen arbitrarily large.  
\end{proof}

\begin{thm}\label{thm:kappa-collapse}\hfill
  \begin{enumerate}[\bf 1.]
  \item 
$W^*_\kappa[F]=W^*[F]$
whenever $F$ has no universal vertex.
\item 
$W^*_\kappa[F]=W^*[F]$ whenever $W^*[F]>3$ and $F$ has no adjacent twins or no non-adjacent twins.
  \end{enumerate}
\end{thm}

\begin{proof}
Part 1 is an immediate consequence of Lemma \ref{lem:padding}.
To establish Part 2, we have to prove that $W^*_\kappa[F]\ge W^*[F]$.
Since with 3 pebbles Spoiler can force playing on connected components,
the assumption $W^*[F]>3$ implies that $W^*[F]=W(G,H)$ for some connected $G$ and $H$
such that $F\sqsubset G$ and $F\not\sqsubset H$.
Assume that $F$ has no adjacent twins.
Consider $G'=G\cdot K_s$ and $H'=H\cdot K_s$ (recall that $\cdot$ denotes the lexicographic product of graphs). 
Note that $G'$ and $H'$ are $s$-connected and observe that $W(G',H')\ge W(G,H)$.
Moreover, $G'$ still contains an induced copy of $F$,
and $H'$ still does not (because if an induced subgraph of $H'$ contains two vertices
from the same $K_s$-part, they are adjacent twins in this subgraph).
Since $s$ can be chosen arbitrarily large, this implies that $W^*_\kappa[F]\ge W(G,H)=W^*[F]$,
as required. If $F$ has no non-adjacent twins, the same argument works with $G'=G\cdot\compl{K_s}$ and $H'=H\cdot\compl{K_s}$.
\end{proof}

An example of a graph to which Theorem \ref{thm:kappa-collapse} is
non-applicable is the diamond. It has universal vertices and both adjacent and non-adjacent twins.

If an $\ell$-vertex pattern graph $F$ has no universal vertex, then $W^*_\kappa[F]=W^*[F]$ and,
therefore, $W^*_\kappa[F]\ge(\frac12-o(1))\ell$ by Theorem \ref{thm:W^*[F]}.
Lemma \ref{lem:padding} works well also for $F$ with few universal vertices.
For example, we have $W^*_\kappa[K_{1,\ell}]\ge W^*[\compl{K_{\ell}}]=W^*[K_{\ell}]=\ell$.
However, if $F$ has many universal vertices, then we need a different approach.

Similarly to $\alice F$, we define $\alicek F$
to be the maximum $k$ such that, for each $s$, there is an $s$-connected graph $H$ with
$H\models \ea {k-1}$ while $F\not\sqsubset H$.

The following relations are easy to prove similarly to Lemma \ref{lem:alice} and Theorem \ref{thm:chi}
(note that, for each fixed $s$, the random multipartite graph $\rturan$ in the proof of Theorem \ref{thm:chi} 
is $s$-connected with high probability).

\begin{lem}\label{lem:chi2}
$W^*_\kappa[F]\ge\alicek F\ge\chi(F)$.
\end{lem}

\begin{thm}\label{thm:W^*_kappa[]}
If $F$ has $\ell$ vertices, then
 $W^*_\kappa[F]>\frac13\ell-\frac43\log_2\ell$.
\end{thm}

\begin{proof}
Denote the number of universal vertices in $F$ by $m$, and let
$F_0$ be obtained from $F$ by removing all these vertices.
By Lemma \ref{lem:padding} and Theorem \ref{thm:W^*[F]},
\begin{equation}
  \label{eq:FF0}
W^*_\kappa[F]\ge W^*[F_0]>\frac12\,(\ell-m)-2\log \ell.
\end{equation}
By Lemma \ref{lem:chi2},
\begin{equation}
  \label{eq:Fchi}
W^*_\kappa[F]\ge\chi(F)\ge m. 
\end{equation}
Combining the bounds \refeq{FF0} and \refeq{Fchi}, we obtain
$$
3\,W^*_\kappa[F]>\ell-4\log \ell,
$$
which implies the bound stated in the theorem.
\end{proof}

\begin{rem}
In contrast to Lemma \ref{lem:compl},
we cannot deduce just from the definitions that $W^*_\kappa[F]=W^*_\kappa[\compl F]$
because the complement $\compl G$ of a highly connected graph $G$ can have low connectivity.
However, the complement of the random multipartite graph $\rturan$ 
is as well highly connected. In view of this fact,
a slight modification of the proof of Theorem \ref{thm:chi} and Lemma \ref{lem:chi2}
gives us also the bound $W^*_\kappa[F]\ge\chi(\compl F)$,
which implies, for example, that 
$W^*_\kappa[P_\ell]\ge(\ell-1)/2$ and $W^*_\kappa[C_\ell]\ge(\ell-1)/2$.
\end{rem}

Finally, we determine the values of $W^*_\kappa[F]$ and $W_\kappa[F]$ for small connected
pattern graphs.
As a particular case of Lemma \ref{lem:chi2}, we have $W^*_\kappa[K_3]=3$ and $W^*_\kappa[K_4]=4$.
According to the discussion in Section \ref{s:intro}, we have $W_\kappa[P_3]\le D_v[P_3]\le2$.
 On the other hand, $W^*_\kappa[P_3]\ge2$
just because there are highly connected graphs with an induced copy of $P_3$
and without it, for example, $K_n\setminus e$ and $K_n$ respectively.

\begin{thm}\hfill
\begin{enumerate}[\bf 1.]
\item 
$W^*_\kappa[K_3+e]=W[K_3+e]=3$;
\item 
$W^*_\kappa[P_4]=W^*[P_4]=3$ and $W_\kappa[P_4]=W[P_4]=4$;
\item 
$W^*_\kappa[F]=W[F]=4$ for all remaining connected $F$ on 4 vertices.
\end{enumerate}
\end{thm}

\begin{proof}
1.  
In the trivial direction, $W^*_\kappa[K_3+e]\le W[K_3+e]=3$,
the equality being established in Theorem~\ref{thm:paw}.
On the other hand, we have
$$
W^*_\kappa[K_3+e]\ge W^*[K_2+K_1]=W^*[\compl{K_2+K_1}]=W^*[P_3]\ge\alice{P_3}=3,
$$
where we use Lemma~\ref{lem:padding}, Lemma~\ref{lem:compl}, and Part~1 of Example~\ref{ex:3-vertex}.

2. We have
$W^*_\kappa[P_4]=W^*[P_4]=3$ by Part 1 of Theorem~\ref{thm:kappa-collapse} and by Theorem~\ref{thm:P_4-WGH}.
Since Lemma~\ref{lem:padding}
easily extends to the relation $W_\kappa[F]\ge W[F_0]$,
we also have $W_\kappa[P_4]=W[P_4]=4$.

3. We have
$W^*_\kappa[C_4]=W^*[C_4]=4$ by Part 1 of Theorem~\ref{thm:kappa-collapse} and by Theorem~\ref{thm:C_4}.

We also have
$W^*_\kappa[K_{1,3}]=\alicek{K_{1,3}}=4$ and $W^*_\kappa[K_4\setminus e]=\alicek{K_4\setminus e}=4$
by the first inequality in Lemma \ref{lem:chi2}. Indeed,
$\alicek{K_{1,3}}=\alicek{K_4\setminus e}=4$ because the $m\times m$ rook's graph
is a strongly regular graph
containing no induced copy of $K_{1,3}$ and no induced copy of~$K_4\setminus e$.
Moreover, the $m\times m$ rook's graph is $(2m-2)$-connected by the following
general fact: The connectivity of a connected strongly regular graph equals its
vertex degree~\cite{BrouwerM85}.
\end{proof}

Note that the equality $D_\kappa[K_4\setminus e]=4$ follows also from the general lower bound
$D_\kappa[F]\ge\frac{e(F)}{v(F)}+2$ proved in \cite[Theorem 5.4]{VZh16},
where $e(F)$ and $v(F)$ denote, respectively, the number of edges and vertices in the graph~$F$.

\section{Trading super-recursively many first-order quantifiers 
for a single monadic one}\label{s:emso}

We now turn to \emph{existential monadic second-order logic}, denoted by \emso,
whose formulas are of the form
\begin{equation}
  \label{eq:emso}
\E X_1\ldots\E X_m\,\Phi,
\end{equation}
where a first-order subformula $\Phi$ is preceded by (second-order) quantification over
unary relations (that is, we are now allowed to use existential
quantifiers over subsets of vertices $X_1,X_2,\ldots$).
The second-order quantifiers contribute in the \emph{quantifier depth} as well as
the first-order ones. Thus, \refeq{emso} has quantifier depth $m$ larger than~$\Phi$.
If a graph property $\classc$ is definable in \emso, the minimum quatifier depth
of a defining formula will be denoted by $\emsod(\classc)$.
Furthermore, we define $\emsod[F]=\emsod(\indsubgr F)$.

It is very well known that \emso
is strictly more expressive than first-order logic (FO for brevity). For example,
the properties of a graph to be disconnected or to be bipartite
are expressible in \emso but not in FO.
We now show that \emso is also much more succinct
than FO, which means that some properties of graphs that are expressible in FO
can be expressed in \emso with significantly smaller quantifier depth.
In fact, this can be demonstrated by considering the properties of
containing a fixed induced subgraph. It turns out that, if we are allowed
to use just one monadic second-order quantifier, the number of
first-order quantifiers can sometimes be drastically reduced.

We will use the known fact that there are graphs definable in FO very succinctly.
According to our notation, $D(\{F\})$ denotes the minimum quantifier depth
of a first-order formula that is true on $F$ but false on every graph non-isomorphic
to $F$. This parameter is called in \cite{PikhurkoV11} the \emph{logical depth
of a graph $F$}. It turns out that the logical depth can be enormously smaller
than the number of vertices in the graph.

Let $v(F)$ denote the number of vertices in~$F$.

\begin{lemC}[Pikhurko, Spencer, Verbitsky \cite{PikhurkoSV06}]\label{lem:psv}
There is no total recursive function $h$ such that
$$
h(D(\{F\}))\ge v(F)
$$
for all graphs~$F$.
\end{lemC}

\begin{thm}\label{thm:emso}
There is no total recursive function $f$ such that
\begin{equation}
  \label{eq:recrel}
f(\emsod[F])\ge D[F]
\end{equation}
for all graphs~$F$. Moreover, this holds true even for the fragment of \emso
where exactly one second-order quantifier is allowed.
\end{thm}

\begin{proof}
Assume that there is a total recursive function $f$ such that \refeq{recrel}
is true for all graphs $F$. Moreover, we can suppose that $f$ is monotonically increasing.

Our goal is to find a contradiction to this assumption. We begin with observing that
\begin{equation}
  \label{eq:MDF}
\emsod[F]\le D(\{F\})+1.
\end{equation}
Indeed, the existence of an induced $F$ subgraph in a graph $G$ can be
expressed in \emso by saying that
$$
\E X\, (G[X]\cong F),
$$
where the assertion ``$G[X]\cong F$'' can be written in first-order logic
by taking a succinct first-order definition of the graph $F$ and relativizing each
quantifier in this definition to the set~$X$.

As a direct consequence of Theorem \ref{thm:W^*[F]}, we also obtain a lower bound
\begin{equation}
  \label{eq:DF}
  D[F]\ge\frac{v(F)}8
\end{equation}
for all $F$ (note that the bound is trivially true if $v(F)\le16$).
Using the inequalities \refeq{MDF}, \refeq{recrel}, and \refeq{DF},
we see that
$$
f(D(\{F\})+1) \ge f(\emsod[F]) \ge D[F] \ge \frac{v(F)}8
$$
for all $F$.
This contradicts Lemma~\ref{lem:psv} by considering the function $h(x)=8f(x+1)$.
\end{proof}

In the proof of Theorem \ref{thm:emso}, we used the relation $\emsod[F]\le D(\{F\})+1$.
It may be curious to note that the value of $\emsod[F]$ can be even smaller than $D(\{F\})$.
For example, for the triangle and the diamond graph we have
$D(\{K_3\})=D(\{K_4\setminus e\})=4$, whereas $\emsod[K_3]=\emsod[K_4\setminus e]=3$.

\section{Conclusion}
\hfill
\que 
As noticed by Mikhail Makarov (personal communication 2018; see also \cite{arxiv}),
the equality $D[K_3+e]=3$ leads to infinitely many examples of graphs $F$ with $D[F]$
strictly less than the number of vertices in $F$. In fact, this follows from
an observation that $D[F_1+F_2] \le D[F_1]+v(F_2)$
for any graphs $F_1$ and $F_2$. Can this upper bound for $D[F]$ or $W[F]$ be further improved?
On the other hand, we only know that $W[F]\ge(\frac12-o(1))\ell$ for all $F$
with $\ell$ vertices. This does not even exclude the possibility that
the time bound $O(n^{W[F]})$ for Induced Graph Isomorphism can be better
than the Ne\v{s}et\v{r}il-Poljak bound $O(n^{(\omega/3)\ell+c})$ 
for infinitely many pattern graphs $F$.

\que 
An example of $F=P_4$ shows that
$W^*[F]$ can be strictly less than $W[F]$ but we  do not know 
how far apart from each other $W^*[F]$ and $W[F]$ can generally be.

\que 
Note that Lemmas \ref{lem:alice} and \ref{lem:alice-lower}
show the following hierarchy of graph parameters:
\begin{equation}
  \label{eq:hierarchy}
(1/2-o(1))\,\ell\le E[F]\le W^*[F]\le W[F]\le D[F]\le\ell,
\end{equation}
where $\ell$ is the number of vertices in $F$.
It seems that we currently have no example separating
the parameters $W[F]$ and $D[F]$ or the parameters $\alice F$ and $W^*[F]$.
An important question is whether or not $\alice F=(1-o(1))\ell$.

\que 
It follows from \refeq{hierarchy} that $W[F]\le(2+o(1))W^*[F]$.
In other terms, in the context of Induced Subgraph Isomorphism, 
the infinitary logic cannot be much more succinct than the standard
first-order logic with respect to the number of variables.
More generally, is it true that $W(\classc)=O(W^*(\classc))$ for all first-order definable graph properties?

\que 
We have checked that $D[F]=W[F]=\ell$ for all $F$ with $\ell\le4$ vertices excepting for
the paw graph and its complement. Since it remains open whether the equality holds true for all larger graphs,
it seems reasonable to examine the pattern graphs on 5 vertices.
If there is a 5-vertex $F$ with $W[F]=4$, the resulting decision procedure for \indsubgr F
would be competitive to (or, at least comparable with) the currently known algorithmic
results for 5-vertex induced subgraphs~\cite{FloderusKLL13,WilliamsWWY15}.

\que 
We have a (rather trivial) example of $F=P_3$ showing that $W_\kappa[F]$
can be strictly smaller than $W[F]$.
Are there other such graphs?

\subsection*{Acknowledgement}
We thank the anonymous referees for careful reading of our paper
and many useful comments.

\end{document}